\documentclass[%
 reprint,
superscriptaddress,
 amsmath,amssymb,
 aps,
pra,
]{revtex4-2}
\usepackage{graphicx}
\usepackage{xcolor}
\usepackage{dcolumn}
\usepackage{bm}
\usepackage{physics}
\usepackage{amsthm}
\usepackage{dsfont}
\usepackage{url}
\usepackage{hyperref}
\usepackage{makecell}
\usepackage{CJKutf8}
\newtheorem{theorem}{Theorem}

\newcommand{\GHZket}{\ket{\mathrm{GHZ}}}

\DeclareMathOperator*{\argmin}{arg\,min}

\newcommand{\NL}{\mathrm{NL}}

\newcommand{\ORt}{\mathrm{OR}_3}
\newcommand{\ORtp}{\mathrm{OR}_3^{\oplus}}
\newcommand{\nand}{\mathrm{NAND}_2}

\newcommand{\ps}{\bar{p}_{S}}

\newcommand{\psc}{\bar{p}_{S}^{c}}

\makeatletter
\newcommand{\thickhline}{%
    \noalign {\ifnum 0=`}\fi \hrule height 1pt
    \futurelet \reserved@a \@xhline
}
\newcolumntype{"}{@{\hskip\tabcolsep\vrule width 1pt\hskip\tabcolsep}}
\makeatother

\begin{document}

\preprint{APS/123-QED}

\title{Non-adaptive measurement-based quantum computation on IBM Q}

\author{Jelena Mackeprang}
\altaffiliation{Present address: Centrum Wiskunde \& Informatica (CWI), Science Park 123, 1098 XG Amsterdam, The Netherlands}
\affiliation{Institute for Functional Matter and Quantum Technologies, University of Stuttgart, 70569 Stuttgart, Germany}
\affiliation{Center for Integrated Quantum Science and Technology (IQST), University of Stuttgart, 70569 Stuttgart, Germany}
\author{Daniel Bhatti}
\affiliation{Institute for Functional Matter and Quantum Technologies, University of Stuttgart, 70569 Stuttgart, Germany}
\affiliation{Center for Integrated Quantum Science and Technology (IQST), University of Stuttgart, 70569 Stuttgart, Germany}
\author{Stefanie Barz}
\affiliation{Institute for Functional Matter and Quantum Technologies, University of Stuttgart, 70569 Stuttgart, Germany}
\affiliation{Center for Integrated Quantum Science and Technology (IQST), University of Stuttgart, 70569 Stuttgart, Germany}

\begin{abstract}
We test the quantumness of IBM’s quantum computer \textit{IBM Quantum System One} in Ehningen, Germany.
We generate generalised n-qubit GHZ states and measure Bell inequalities to investigate the n-party entanglement of the GHZ states.
The implemented Bell inequalities are derived from non-adaptive measurement-based quantum computation (NMQC), a type of quantum computing that links the successful computation of a non-linear function to the violation of a multipartite Bell-inequality.
The goal is to compute a multivariate Boolean function that clearly differentiates non-local correlations from local hidden variables (LHVs).
Since it has been shown that LHVs can only compute linear functions, whereas quantum correlations are capable of outputting every possible Boolean function it thus serves as an indicator of multipartite entanglement.
Here, we compute various non-linear functions with NMQC on IBM’s quantum computer IBM Quantum System One and thereby demonstrate that the presented method can be used to characterize quantum devices.
We find a violation for a maximum of seven qubits and compare our results to an existing implementation of NMQC using photons.

\end{abstract}

\maketitle


\section{\label{sec:intor} Introduction}

Commercially available quantum computers (QCs) have arrived in the NISQ (\textit{noisy intermediate-scale quantum}) era~\cite{Preskill2018}. 
Equipped with $10$s to $100$s of of noisy qubits, these devices already enable the implementation of quantum operations and thus basic quantum algorithms~\cite{Bharti2022}. Despite the lack of error correction, algorithms and techniques adapted to the strengths and shortcomings of the computers could facilitate non-classical computation within the near future.
To compare the performance of the large range of different quantum devices and to find the best-suited QC for a specific problem, benchmarking, i.e., reproducibly measuring the performance of quantum devices, becomes especially important~\cite{Eisert2020}.

To be independent of the architecture and capture the complexity of quantum machines, benchmarking protocols go beyond comparing the various hardware characteristics~\cite{Moll2018,Lubinski2021}.
The goal is to find protocols that give maximal information about the performance of a quantum device~\cite{Bharti2022}.
Examples for such hardware benchmarks are randomised benchmarking~\cite{Helsen2022}, cross-entropy benchmarks~\cite{Arute2019etal}, or the quantum volume~\cite{Moll2018,Cross2019}. Besides that, application benchmarks exist which test the performance of NISQ devices based on their execution of different applications or algorithms and help one to understand how good QCs can deal with different tasks~\cite{Bharti2022,Lubinski2021}.

One fundamental type of application that can be used to benchmark quantum devices is the generation and verification of entanglement~\cite{Alsina2016,Swain2019,Huang2020,Gonzalez2020,Baumer2021,Yang2022,Wang2018,Mooney2021AQT,Wei2020,Mooney2021}. 
To this aim, various tests of multipartite entanglement have been implemented, e.g. utilising Mermin inequalities~\cite{Alsina2016,Swain2019,Huang2020,Gonzalez2020} or multiparty Bell inequalities~\cite{Baumer2021,Yang2022}, but also measuring the entanglement between all connected qubits in a large graph state~\cite{Wang2018,Mooney2021AQT}. In the case of Greenberger-Horne-Zeilinger (GHZ) states a feasible method to estimate the fidelity has been derived and implemented to verify the state generation of large numbers of qubits~\cite{Wei2020,Mooney2021}.

In this work, we make use of a method called non-adaptive measurement-based quantum computation (NMQC) to 
characterise an IBM QC with 27 superconducting qubits.

The goal in NMQC is to compute a multivariate function. While local hidden variables (LHVs) can only output linear functions, quantum correlations can compute all Boolean functions. The success of such a computation can be related to the violation of a (generalised) Bell inequality and proves the advantage over classical resources~\cite{Hoban2011IOP}. So far, binary NMQC has been implemented with four-photon GHZ states~\cite{Demirel2021}.
Here, we use GHZ states on an IBM QC to implement NMQC with more than four qubits. This allows us to test the quantum correlations of the generated GHZ states and therefore the non-classicality of the respective IBM quantum computer.

In particular, we implement NMQC for one two-bit function, three three-bit functions, and one four-bit, one five-bit, and one six-bit function on the superconducting quantum computing system IBM Quantum System One (QSO) and demonstrate that it exhibits multipartite entanglement. For qubit numbers lower or equal to five, we utilise quantum readout error mitigation~\cite{Maciejewski2020} to reduce noise from local measurement errors. For higher qubit numbers, we utilise the error mitigation tools provided by Qiskit~\cite{Qiskitetal}. We demonstrate violations of the associated Bell inequalities for up to seven qubits, which indicates the non-classical properties of the quantum computing system.

\section{\label{sec:BG}Background}
\subsection{\label{sec:BG:NMQC}NMQC}

First, let us briefly describe the general idea of NMQC (for a detailed overview over the procedure see Fig.~\ref{fig:NMQCscheme} and, e.g. Refs.~\cite{Hoban2011IOP,Mackeprang2022}): Starting from a classical $n$-bit input string $x = (x_{0},x_{1},\ldots,x_{n-1}) \in \{0,1\}^n$, which is sampled from a probability distribution $\xi(x)$, the goal is to compute any multivariate Boolean function $f: \{0,1\}^n \rightarrow \{0,1\}$. For this, one has access to a restricted classical computer limited to addition mod~2, which can be used for classical pre- and post-processing (see Fig.~\ref{fig:NMQCscheme}).
The core of NMQC is embodied by non-adaptive measurements on an $l$-qubit resource state, with $l\leq n$.
It has been shown that if the measurement statistics are described by local hidden variables (LHVs)~\cite{Bell1964}, i.e. one uses a classical resource state, the output of NMQC is restricted to linear functions. As the pre-processor is already capable of outputting linear functions LHVs thus do not ``boost'' the pre-processor's computational power~\cite{Hoban2011IOP}.

In contrast, non-local quantum correlations, can elevate the pre-processor to classical universality. The generalised $l$-qubit GHZ state
\begin{equation} \label{eq:2GHZ}
	\GHZket= \frac{1}{\sqrt{2}}\left(\ket{0}^{\otimes l}+\ket{1}^{\otimes l}\right),	\end{equation}
enables the computation of  \emph{all} functions $f: \{0,1\}^n \rightarrow \{0,1\}$ with at most $l=2^n-1$ qubits. The computation of a non-linear function 
requires \emph{non-locality}~\cite{Brunner2014} and can be seen as a type of GHZ paradox~\cite{Hoban2011IOP}.
Thus, the successful execution of NMQC demonstrates non-locality. Note that in our case the non-locality is realised by single-qubit measurements on an multipartite entangled state and we can use the two terms non-locality and multipartite entanglement interchangeably.

In general, it can be shown that the average success probability $\ps = p(z=f(x))$, i.e. the probability that the output $z$ is identical to the value of the target function $f(x)$, is related to a normalised Bell inequality $\beta$ with a classical (LHV) bound $\beta_c$ and a quantum bound $\beta_q$ \cite{Hoban2011IOP}:
\begin{align} 
	2\ps -1 = \beta &= \sum_{x} (-1)^{f(x)}\xi(x) E(x) \leq \begin{cases}  \beta_c \\ \beta_q
	\end{cases} . \label{eq:psu}
\end{align}
The expectation values are defined:
\begin{equation}\label{eq:E2def}
	E(x) = p(z=0|x)-p(z=1|x),
\end{equation}
where $p(z=k|x)$ is the probability that $z$ is equal to $k$ for the input $x$.

It has been shown that the GHZ state always maximally violates the given Bell inequalities and minimises the number of required qubits for a violation~\cite{Werner2001,Zukowski2002}. It is thus optimal for NMQC~\cite{Hoban2011IOP} and we will use it as a resource in the following investigations on IBM QSO.
 
	\begin{figure}[b]
		\centering
		 \includegraphics[width=.8\linewidth]{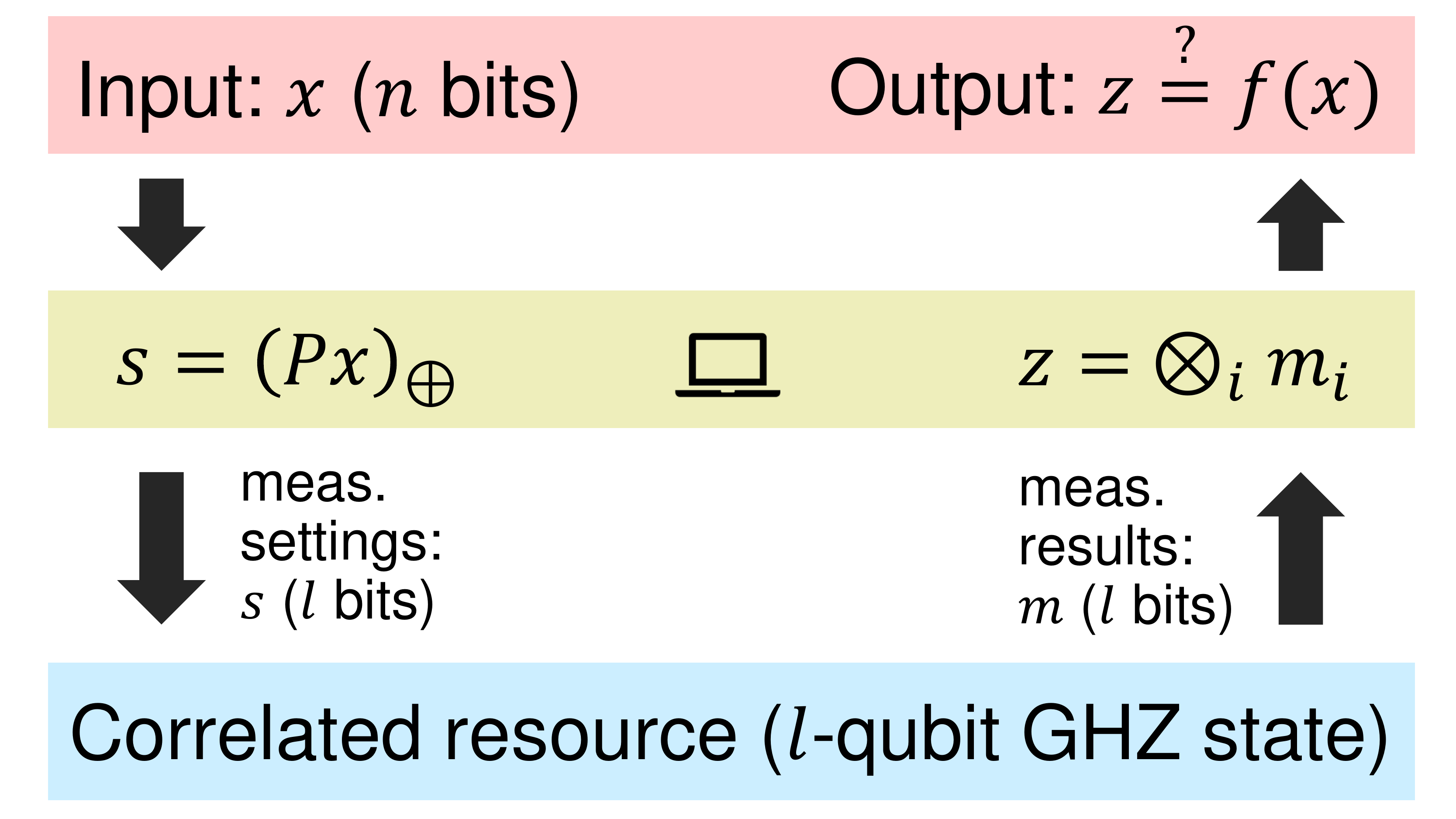}
			\caption{The figure shows the general scheme of NMQC. At the beginning, an input string $x \in \{0,1\}^n$ is sent to the parity computer, which in turn computes the bit string $s \in \{0,1\}^l$. This restricted computation can be seen as a matrix vector multiplication: $s=(Px)_{\oplus}$, where $P$ is an $l$-by-$n$ binary matrix and $\oplus$ denotes that the matrix vector product is evaluated w.r.t. mod~2 operations. Each bit $s_{i}=0,1$ in $s$ now determines the settings for the measurement on the $i$th qubit of the $l$-qubit resource state. For each subsystem, there are two measurement operators $\hat{m}_{i}(s_{i})$, one can choose from (here $\hat{m}_{i}(0)\equiv X$ and $\hat{m}_{i}(1)\equiv Y$, with $X$ and $Y$ denoting the Pauli operators). Each measurement yields one of two possible measurement results $M_i \in \{-1,1\}$, which can be mapped to bits $m_i \in \{0,1\}$: $ M_i = (-1)^{m_i} $. The measurements are performed on a correlated $l$-qubit resource and the measurement results $m_{i} \in \{0,1\}$ are summed up by the parity computer: $z \equiv \bigoplus_i m_i$. Finally, if $z=f(x)$ for this input $x$, the computation was successful. Note that if $z = f(x)$ for every $x$, we say that an NMQC scheme is \textbf{deterministic}. The figure has been adapted from~\cite{Mackeprang2022}.}
			\label{fig:NMQCscheme}
	\end{figure}

\subsection{Tested Functions and Bell inequalities}\label{sec:BellInequ}

The NMQC computations for the four-qubit GHZ state presented in this work result from the two-variate function:
\begin{equation}
	\mathrm{NAND}_2(x)= x_0x_1 \oplus 1 \label{eq:NMQCIBM:funcs4},
\end{equation}
and the three three-variate functions:
\begin{align}
	h_3(x)&=   x_0x_1\oplus x_0x_2 \oplus x_1x_2 \oplus x_0 \oplus x_1 \oplus x_2 , \label{eq:NMQCIBM:funcs1} \\
	\mathrm{OR}_3(x)  &= x_0 \lor x_1 \lor x_2 , \label{eq:NMQCIBM:funcs2} \\	
	\mathrm{OR}_3^{\oplus}(x)&=\mathrm{OR}_3(x)  \oplus x_0x_2 , \label{eq:NMQCIBM:funcs3}
\end{align}
where $\lor$
is the logical OR operator and $\oplus$ denotes addition mod~2. Note that the same functions have been used to implement NMQC using a four-photon GHZ state in Ref.~\cite{Demirel2021}.

In the case of the two-bit function $\mathrm{NAND}_2(x)$, we use a uniform probability distribution $\xi(x)= \frac{1}{4}$,
which yields the Bell inequality [see Eq.~(\ref{eq:psu})]:
\begin{align} \label{eq:NMQCIBM:nand2Bell}
	\beta_{\mathrm{NAND}_2} &= \frac{1}{4} \left[ - E((0,0)) - E((1,0)) \right. \nonumber \\
	&\phantom{{}={}} \left. - E((0,1))  +  E((1,1)) \right]  \leq \left\{\begin{array}{l}
		\beta_{c} = 1/2\\
		\beta_{q} =1 
	\end{array} . \right.
\end{align}

The relation between the measurement settings and the measurements, i.e., $\hat{m}_{i}(s_{i}) = X/Y$ for $s_{i}=0/1$ ($i\in\{0,1,2,3\}$)
allows us to rewrite the Bell inequality~(\ref{eq:NMQCIBM:nand2Bell}) in terms of the four measurements:
\begin{align} \label{eq:NMQCIBM:nand2Bell2}
	\beta_{\mathrm{NAND}_2} &= \frac{1}{4} \left< -XXYY - YXXY \right. \nonumber \\
	&\phantom{{}={}} \left. - XYXY + YYYY \right>  \leq \left\{\begin{array}{l}
		\beta_{c} = 1/2\\
		\beta_{q} =1 
	\end{array} \right.,
\end{align}
where we additionally made use of the following relation between measurement settings and the input bits $x_i$:
\begin{equation}
	s_0 = x_0, \quad s_1 = x_1, \quad s_2 = x_0 \oplus x_1 \oplus 1 , \quad s_3=1.
\end{equation}

In the same manner one finds the Bell inequalities for the three three-variate functions given in Eqns.~(\ref{eq:NMQCIBM:funcs1}-\ref{eq:NMQCIBM:funcs3})~\cite{Demirel2021}. The inequalities and the respective pre-processing implemented in Sec.~\ref{sec:results} are shown in Table~\ref{tab:Inequalities}.

To perform NMQC for five- to seven-qubit GHZ states we use the generalisation of $h_3(x)$, namely $h_k(x)$, for $k=4,\,5$ and 6:
\begin{equation}\label{eq:hk}
	h_k(x) = \bigoplus_{i=0}^{k-2}x_i \left(\bigoplus_{j=i+1}^{k-1} x_j\right)  \oplus \bigoplus_{i=0}^{k-1} x_i.
\end{equation}

\begin{table*}[t]
	\centering
	\begin{tabular}{lcccccccc"cc} 
		$x$ & $(0,0,0)$ & $(1,0,0)$ & $(0,1,0)$ & $(0,0,1)$ & $(1,1,0)$ & $(1,0,1)$ & $(0,1,1)$ & $(1,1,1)$  \\		
		\hline
		$\left< \hat{m}_{0}(s_{0})\hat{m}_{1}(s_{1})\hat{m}_{2}(s_{2})\hat{m}_{3}(s_{3}) \right>$ & $\left< XXXX \right>$ & $\left< YXXY \right>$ & $\left< XYXY\right>$ & $\left< XXYY \right>$ & $\left< YYXX \right>$ & $\left< YXYX \right>$ & $\left<XYYX \right>$ & $\left< YYYY \right>$ & $\beta_{c}$ & $\beta_{q}$ \\		
		\thickhline
		$(-1)^{\mathrm{OR}_3(x)} \xi(x)$ & $+3/10$ & $-1/10$ & $-1/10$ & $-1/10$ & $-1/10$ & $-1/10$ & $-1/10$ & $-1/10$ & $4/10$ & $8/10$ \\
		\hline 
		$(-1)^{\mathrm{OR}_3^{\oplus}(x)} \xi(x)$ & $+1/16$ & $-3/16$ & $-3/16$ & $-1/16$ & $-3/16$ & $+1/16$ & $-3/16$ & $+1/16$ & $9/16$ & $14/16$ \\
		\hline
		$(-1)^{h_3(x)} \xi(x)$ & $+1/8$ & $-1/8$ & $-1/8$ & $-1/8$ & $-1/8$ & $-1/8$ & $-1/8$ & $+1/8$ & $1/2$ & $1$
	\end{tabular}
	\caption{Bell inequalities for the three three-variate functions $\mathrm{OR}_3(x)$, $\ORtp(x)$ and $h_3(x)$ given in Eqns.~(\ref{eq:NMQCIBM:funcs1}-\ref{eq:NMQCIBM:funcs3}) and implemented in Sec.~\ref{sec:results}. The first row shows all possible three-bit inputs, while the second row gives the respective measurements after pre-processing. The pre-processing used for all three functions is: $s_0 = x_0$, $s_1 = x_1,$ $s_2 = x_2,$ $s_3 = x_0 \oplus x_1 \oplus x_2$. $\hat{m}_{i}(s_{i}) = X/Y$ for $s_{i}=0/1$ ($i\in\{0,1,2,3\}$). Rows 3-5 show the resulting prefactors of the different measurement results in the Bell inequalities. The classical (quantum) bound $\beta{c}$ ($\beta{q}$) of each Bell inequality are presented in the two separate columns on the right.}\label{tab:Inequalities}
\end{table*}

For any $k$ the sampling distribution is uniform, i.e. $\xi(x)=1/2^k$, and the pre-processing is given by:
 \begin{equation}\label{eq:preprohk}
	s_i	= \begin{cases}
		x_i & i=0,\ldots,k-1\\
		\bigoplus_{j=0}^{k-1} x_j & i=k
	\end{cases} .
\end{equation}

The Bell inequality induced by $h_k(x)$ and defined by the pre-processing~\eqref{eq:preprohk} and the uniform sampling distribution has the quantum bound $q=1$. 
This can be seen by explicitly computing all expectation values.
The classical bounds can either be found numerically or inferred from the connection between NMQC and classical Reed-Muller error-correcting codes, as pointed out in~\cite{Raussendorf_2013}. They are equal to $c=2^{\frac{-k}{2}}$ for even $k$ and $2^{-\left(\frac{k-1}{2}\right)}$ for odd $k$. We elaborate on this in Appendix~\ref{app:proofs} and further show that in order to compute the $k$-bit function $h_k(x)$ with NMQC, one only requires $k+1$ qubits.

\section{NMQC on IBM Quantum System One}

IBM QSO in Ehningen, Germany, is a 27-qubit QC, which we used to run NMQC for up to seven qubits.
Testing a possible violation of Bell inequalities for different qubit configurations of the QC allows for a characterisation of the whole QC or a subset of qubits.

The QC's architecture is shown in Fig.~\ref{fig:ehningensubgraphs}, where each qubit (vertex) is marked by its \emph{physical qubit number} and edges denote physical connections between qubits. Here, physical connection means that two qubits are directly coupled, which allows for a direct implementation of two-qubit gates between those qubits. In the following, when mentioning the physical qubit numbers, we refer to the numbering depicted in Fig.~\ref{fig:ehningensubgraphs}. At the time of the experiment, the quantum computer contained a Falcon r5.11 processor and its backend version was 3.1.9.

\begin{figure}[b]
	\centering
	\includegraphics[width=.9\linewidth]{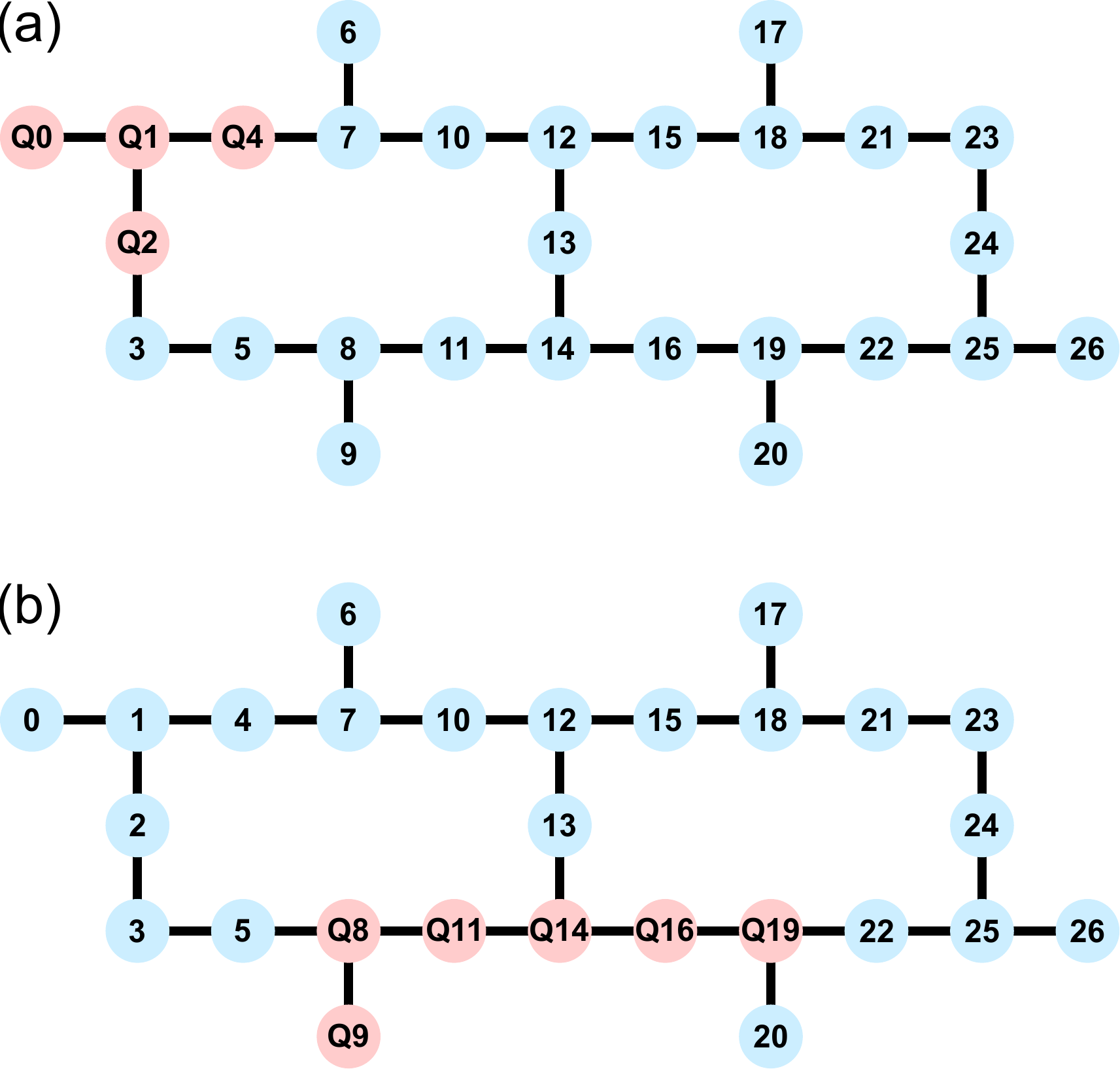}
	\caption{Architecture of IBM Quantum System One for different examples of $l$-qubit configurations. The nodes indicate qubits, marked by the \emph{physical} qubit numbers and the edges denote which ones are physically connected. The qubit configurations are marked in red and the respective qubits are labelled by $Q_k$, where $k$ is the physical qubit number. (a) 4-qubit configuration 0-1-2-4. (b) 6-qubit configuration 8-9-11-14-16-19.}
	\label{fig:ehningensubgraphs}
\end{figure}

We perform two experiments on IBM QSO: (i) In Sec.~\ref{sec:resultsA}, the physical qubits are chosen by Qiskit and the quantum circuit is optimised by Qiskit, and (ii) in Sec.~\ref{sec:resultsC} the physical qubits are chosen manually and the quantum circuit is optimised by our own method (see Sec.~\ref{sec:CreationGHZ}). In both experiments the goal is to generate generalised GHZ states [see Eqn.~\eqref{eq:2GHZ}] as a resource to perform NMQC. 

While in (i) we only test a single configuration, i.e. the one chosen by Qiskit, in (ii) we generate and test every possible $l$-qubit configuration, where $l$ is the number of qubits. By ``qubit configuration'', we mean the collection of $l$ physical qubits that are physically connected in the quantum computer (see Fig.~\ref{fig:ehningensubgraphs}). For each tested Bell inequality in (ii) we then average over all measured bounds for the distinct distributions to determine the measured bound of the whole QC.

\subsection{Creation of the GHZ state}
\label{sec:CreationGHZ}

The scheme used to generate the multi-qubit GHZ states in the first experiment (see Sec.~\ref{sec:resultsA}) follows an easily scalable manner~\cite{Mooney2021}, consisting of a single Hadamard ($H$) gate and $n-1$ CNOT gates (see Fig.~\ref{fig:GHZScheme}). Then, Qiskit chooses the mapping of the virtual to the physical qubits and optimises the quantum circuit according to its highest optimisation level. 

In the second experiment (see Sec.~\ref{sec:resultsC}), where we average over all possible qubit configurations, the qubit onto which the Hadamard gate acts is the one with the largest numbers of neighbours in the configuration and the one with the smallest readout-error rate. The readout-error rates are obtained form the backend's calibration data which is updated before every NMQC run. If the calibration data changed during an NMQC run, the measured data was discarded and the run repeated. The CNOT gates are arranged in such a way that as many as possible of them can be carried out simultaneously, which minimises the circuit depth~\cite{Mooney2021}. Note that CNOT gates are only applied between physically connected qubits.

\begin{figure}[b]
	\centering
	\includegraphics[width= .7\linewidth]{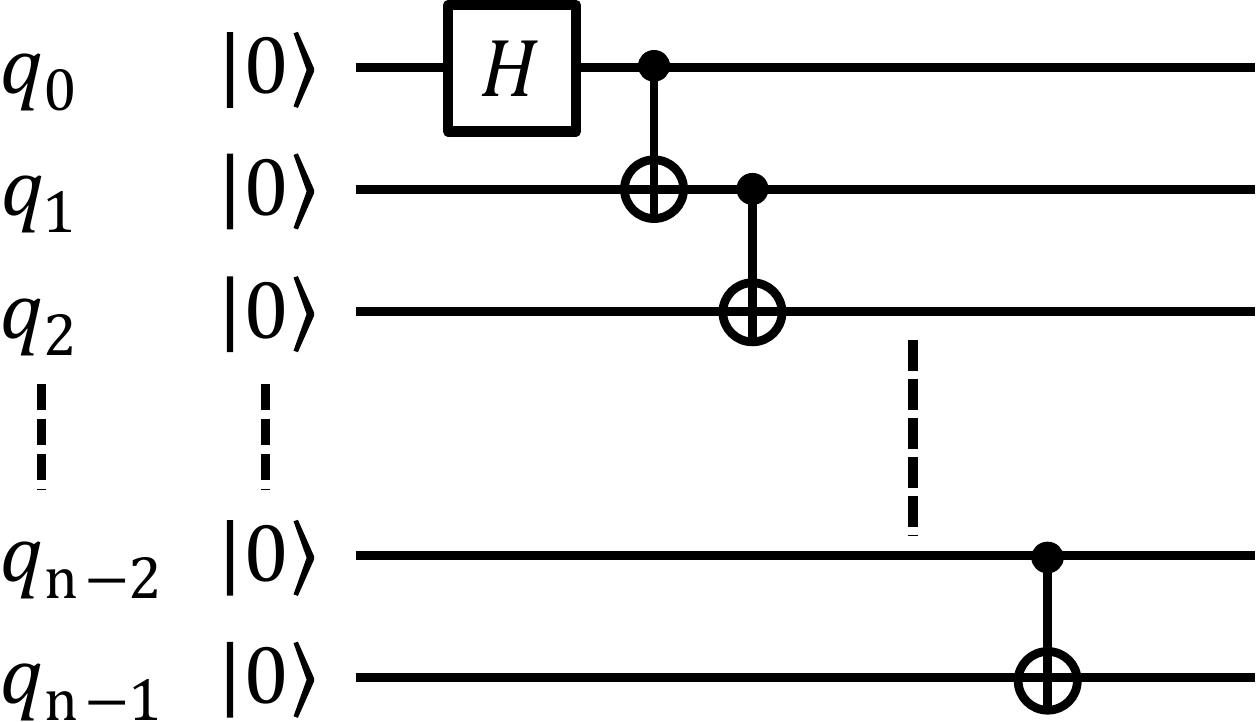}
	\caption{Theoretical scheme for the creation of a GHZ state. First a Hadamard ($H$) is applied on the virtual qubit $q_0$. Then $n-1$ CNOT gates are carried out between the qubits $q_i$ (control) and $q_{i+1}$ (target), where $i$ runs from $i=0$ to $i=n-1$.}
	\label{fig:GHZScheme}
\end{figure}

\subsection{Error mitigation}
\label{sec:ErrorMitigation}

We post-process the measured data for up to five qubits, using the quantum readout error mitigation (QREM)~\cite{Maciejewski2020}. This method has been used, for example, in~\cite{Mooney2021}, where it has led to considerable improvements in the fidelity of a generated multi-qubit GHZ state. It aims at mitigating readout errors, which are errors during the measurement of the state of a single qubit
and the main assumption is that these measurement errors are \emph{local}. We explain the details in Appendix~\ref{sec:QREM}.

To improve the results for NMQC
using six and seven qubits, QREM seems to be insufficient. In fact, we observed a negative effect on the measured bounds and thus switch to the measurement error mitigation (MEM) provided by Qiskit~\cite{Qiskitetal}.
In contrast to QREM Qiskit's MEM does not assume measurement errors to be local but global. This means that instead of $n$ $2\times 2$ calibration matrices $A_i$ one needs to determine a single $2^{n}\times 2^{n}$ calibration matrix $A$ by preparing and measuring all $2^{n}$ basis states.

\begin{figure*}[t]
	\centering
	\includegraphics[width= 0.99 \linewidth]{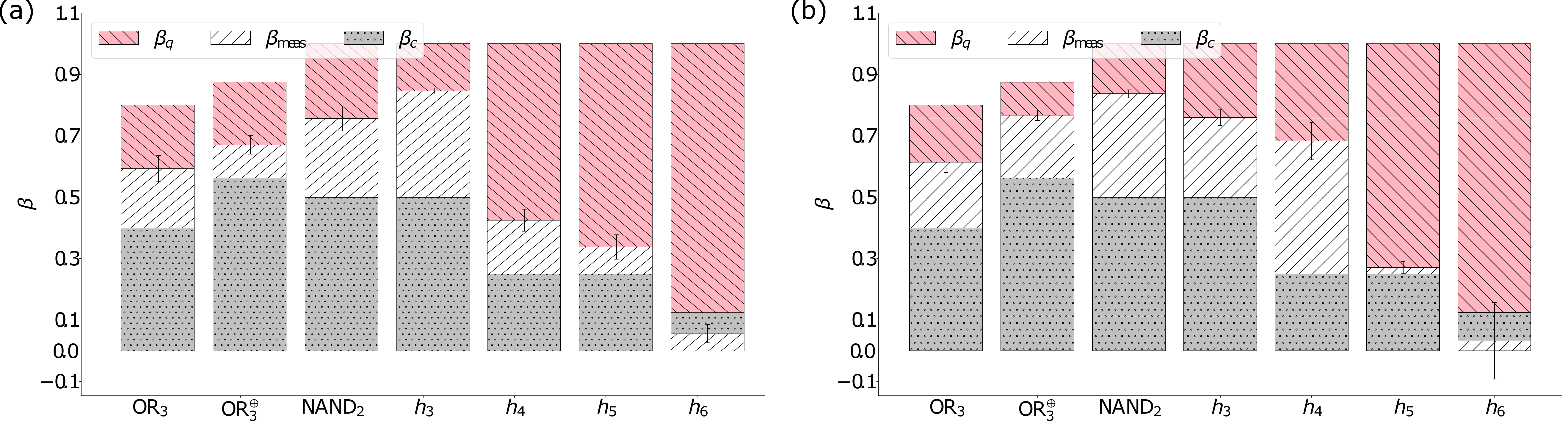}
	\caption{Measured bounds of the Bell inequalities, averaged over $70$ runs with $1000$ shots per circuit, induced by the functions $\ORt(x)$, $\ORtp(x)$, $\nand(x)$ and $h_k(x)$ for $3 \leq k \leq 6$ and their standard deviations for (a) optimisation level 3 and the ``dense'' layout method and (b) optimisation level 3 and the ``noise adaptive'' layout method. The red (diagonally striped \textbackslash \textbackslash)/grey (dotted) bars denote the theoretically achievable quantum/classical bounds and the white (diagonally striped //) bars stand for the measured values. The exact values for all measured bounds are listed in Table~\ref{tab:Data}.}
	\label{fig:3denseplusnoise}
\end{figure*}
 
\section{Results}
\label{sec:results}
 
Here, we present the average values for the violations of the associated Bell inequalities of NMQC listed above. We start with the first experiment, in which the physical qubits are chosen by Qiskit (Sec.~\ref{sec:resultsA}). There, the circuits for NMQC were all transpiled using the option ``optimisation level 3", i.e. heavy optimisation including noise-adaptive qubit mapping and gate cancellation~\cite{Qiskitetal}. We differentiate between two sub-experiments: one, where the circuits were transpiled using the option ``layout\_method=dense'', which chooses the most connected subset of qubits with the lowest noise
and one, where the circuits were transpiled using the option ``layout\_method=noise\_adaptive'', which tries to map the virtual to physical qubits in a manner that reduces the noise~\cite{Qiskitetal}.

In the second experiment we choose the qubits manually (Sec.~\ref{sec:resultsC}), testing all possible qubit configurations to generate the $n$-qubit GHZ states and perform NMQC.

\subsection{Transpilation optimisation level 3}
\label{sec:resultsA}

Fig.~\ref{fig:3denseplusnoise} shows the measured bounds of the Bell inequalities 
for optimisation level 3 and two different layout methods.
Each Bell inequality was tested in $70$ separate runs, where in each run every circuit, induced by the respective function, has been executed $1000$ times, i.e. $1000$ runs.
One can see that for both methods all measured values, except for $h_6(x)$, i.e. seven qubits, are above the classical bounds which translates to a quantum advantage in the associated NMQC games, even when taking into consideration the standard deviations determined from the $70$ runs. This, in turn, means that the quantum average success probability of the probabilistic NMQC games is higher than the LHV one, indicating multipartite entanglement. 
For this experiment one can say that the performance of both methods ``dense'' and ``noise adaptive'' provided by Qiskit was similar.

\subsection{Transpilation optimisation level 0 and error mitigation}
\label{sec:resultsC}

Fig.~\ref{fig:0wemit} (a) shows the measured bounds of the Bell inequalities averaged over all possible qubit configurations for four ($\ORt(x)$, $\ORtp(x)$, $\nand(x)$ and $h_3(x)$), five ($h_4(x)$), six ($h_5(x)$), and seven ($h_6(x)$) qubits and the mitigated bounds improved by error correction.
The error correction techniques applied are QREM (four and five qubits) and Qiskit's integrated MEM (six and seven qubits) (see Sec.~\ref{sec:ErrorMitigation}).
For every qubit configuration there is exactly one NMQC run with $1000$ shots per circuit~\footnote{If the calibration data of the backend had changed during the NMQC run, the data was discarded and the run repeated. Before every run, the data needed for the error mitigation was generated.}.
In Fig.~\ref{fig:0wemit} (b) we show the measured bounds of the qubit configuration, which produced the highest violation (exact values and the physical qubit numbers are shown in Table~\ref{tab:Data}).

 \begin{figure*}[t]
 	\centering
 	\includegraphics[width= 0.99 \linewidth]{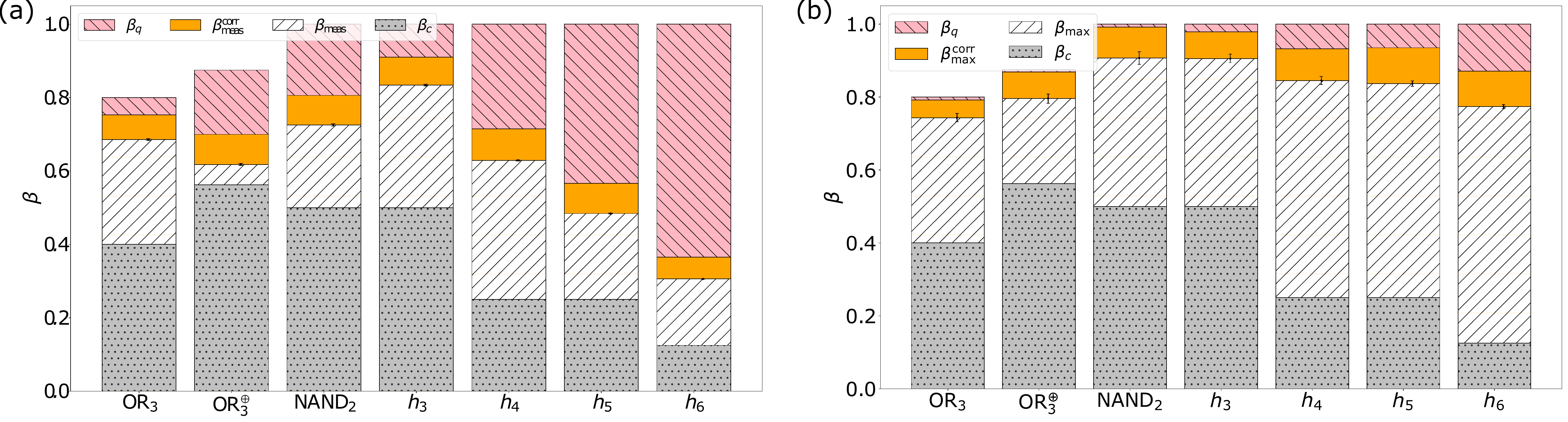}
 	\caption{(a) Measured average bounds of the Bell inequalities induced by the functions $\ORt(x)$, $\ORtp(x)$, $\nand(x)$ and $h_k(x)$ for $3 \leq k \leq 6$ and their standard deviations for optimisation level 0 and the mitigated bounds. The red (diagonally striped \textbackslash \textbackslash)/grey (dotted) bars denote the theoretically achievable quantum/classical bounds and the white (diagonally striped //) bars stand for the measured values. The orange (plain) bars denote mitigated bounds.
	(b) Measured and mitigated bounds of the qubit configuration which produced the highest values, induced by the same functions for optimisation level 0. The exact values for all measured bounds as well as the qubit configurations are listed in Table~\ref{tab:Data}.}
 	\label{fig:0wemit}
 \end{figure*}

One can see from the plots and the data [see Fig.~\ref{fig:0wemit} (a) and (b) and Table~\ref{tab:Data}] that not only the configurations, which produced the highest values, but also the averaged results are significantly higher than the classical bounds of the respective Bell inequalities for any tested function and number of qubits. Especially for more than three qubits the averaged values are higher than in the case of letting Qiskit choose the qubit configuration (see Sec.~\ref{sec:resultsA}). In the case of the single configurations one should keep in mind that these results only express a single run (see Appendix~\ref{sec:IndividualMeasuredBounds}).

The presented error margins correspond to the $99\%$ confidence intervals of the measured values with respect to $1000$ bootstrapped samples for each function except for $h_6(x)$, where we used $100$ bootstrapped samples. We chose bootstrapping~\cite{Efron1994} instead of sampling at different moments in time as the performance of the quantum processor varied considerably. Due to the heavy bias introduced by the optimization procedure used in error mitigation, we omitted error margins as statistical errors are not a meaningful measure in this situation.

\begin{table*}[t]
	\centering
	\begin{tabular}{c|c|c |c |c|c|c|c|c|c|c} 
		& $\ORt(x)$ & $\ORtp(x)$ & $\nand(x)$ & $h_3(x)$ & $h_4(x)$ & $h_5(x)$ & $h_6(x)$ & \makecell{opt.\\ level} & \makecell{layout} & Sec. \\		
		\hline\hline
		$\beta_{\mathrm{meas}}$ &$0.593\pm 0.042$ & $0.670 \pm 0.030$ & $0.757 \pm 0.040$ & $0.846 \pm 0.010$ &$0.426 \pm 0.036$ &$0.339 \pm 0.040$ & $0.056 \pm 0.030$ & 3 & dense & \ref{sec:resultsA} \\
		\hline\hline
		$\beta_{\mathrm{meas}}$ &$0.614\pm 0.034$ & $0.767 \pm 0.017$ & $0.837 \pm 0.013$ & $0.760 \pm 0.026$ &$0.683 \pm 0.026$ &$0.271 \pm 0.061$ & $0.032 \pm 0.124$ & 3 & noise & \ref{sec:resultsA} \\
		\hline\hline  
		$\beta_{\mathrm{meas}}$ &	$0.686 \pm 0.002$ & $0.618 \pm 0.003$  & $0.725 \pm 0.004$ & $0.834 \pm 0.003$ &$0.629 \pm 0.002$ & $0.484 \pm 0.001$& $0.306 \pm 0.001$ & 0 & all & \ref{sec:resultsC}  \\   \hline
		$\beta_{\mathrm{meas}}^{\mathrm{corr}}$ &	$0.753 $ & $0.700 $  & $0.806 $ & $0.911 $ &$0.715 $ & $0.566$ & $0.366 $ & 0 & all & \ref{sec:resultsC} \\
		\hline\hline
		$\beta_{\mathrm{max}}$ & $0.743 \pm 0.012$ & $0.796 \pm 0.013$ & $0.907 \pm 0.017$ & $0.906 \pm 0.012$ & $0.845 \pm 0.011$ & $0.837 \pm 0.008$ & $0.774 \pm 0.006$ & 0 & best & \ref{sec:resultsC} \\ \hline 
		$\beta_{\mathrm{max}}^{\mathrm{corr}}$ &0.792 & 0.869 & 0.992 & 0.979 &0.932 & 0.935 & 0.871  & 0 & best & \ref{sec:resultsC} 
	\end{tabular}
	\caption{Measured averaged values and the standard deviations for the bounds of the Bell inequalities induced by the NMQC games listed in Sec.~\ref{sec:BellInequ}. Results are shown for different Qiskit optimisation levels, Qiskit layout methods and manually chosen qubit mapping, without ($\beta_{\mathrm{meas}}$) and with ($\beta_{\mathrm{meas}}^{\mathrm{corr}}$) error mitigation. For the layout method ``best'' the exact qubit configurations are: $1-2-4-7$ ($\ORt(x)$), $0-1-2-3$ ($\ORtp(x)$), $16-19-14-22$ ($\nand(x)$), $24-18-21-23$ ($h_3(x)$), $10-12-13-14-16$ ($h_4(x)$), $4-7-10-12-13-15$ ($h_5(x)$), $11-12-13-14-16-19-20$ ($h_6(x)$).}\label{tab:Data}
\end{table*}

\subsection{Comparison to Photonic NMQC}

In Ref.~\cite{Demirel2021} binary NMQC has been implemented using four-photon GHZ states, testing the functions $\ORt(x)$, $\ORtp(x)$, $\nand(x)$ and $h_3(x)$. Here, we compare our results using IBM QSO to the photonic results.

In Fig.~\ref{fig:Photons} a) we show the measured bounds from the photonic experiments and the respective standard deviations. For a better comparison, we calculate the difference between these values and the results presented in Sec.~\ref{sec:resultsC}, with and without error mitigation, i.e. $\Delta\beta_{\text{meas/max}} = \beta_{\text{meas/max}}(\text{photons}) - \beta_{\text{meas/max}}(\text{QSO})$ and $\Delta\beta_{\text{meas/max}}^{\text{corr}} = \beta_{\text{meas/max}}^{\text{corr}}(\text{photons}) - \beta_{\text{meas/max}}^{\text{corr}}(\text{QSO})$. In Fig.~\ref{fig:Photons} b) we plot the difference to the measured bounds averaged over all qubit configurations (see Fig.~\ref{fig:0wemit} a)) and in Fig.~\ref{fig:Photons} c) we plot the difference to the measured bounds of the qubit configuration which produced the highest values (see Fig.~\ref{fig:0wemit} b)).

We find that the photonic values are higher than the uncorrected results using IBM QSO comparing to both the averaged bounds and the highest bounds (except $\Delta\beta_{\text{max}}$ for OR$_3$). Using error mitigation the averaged values come closer to the photonic results but only exceed them in the case of OR$_3$. Only when applying error mitigation to the highest values produced by a single qubit configuration the photonic results are exceeded for every function. Additionally, one has to take into account that the values measured on IBM QSO strongly vary depending on the configuration and the time of execution (see Appendix~\ref{sec:IndividualMeasuredBounds}). The possibility to go to larger numbers of qubits remains a big advantage of IBM QSO.

 \begin{figure*}[t]
 	\centering
 	\includegraphics[width= 0.99 \linewidth]{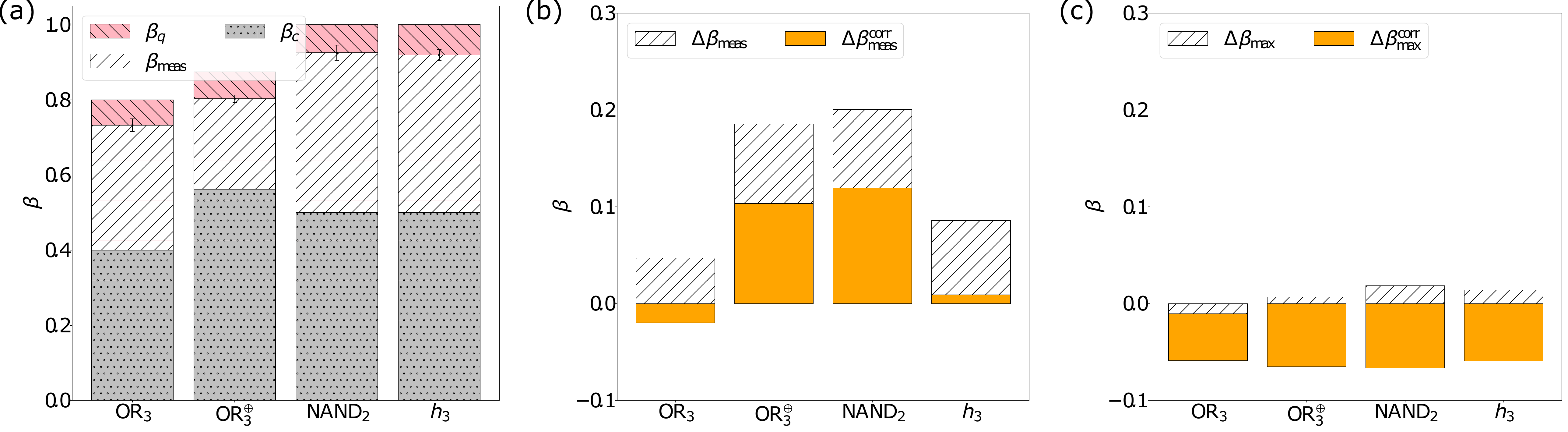}
 	\caption{(a) Measured average bounds of the Bell inequalities induced by the functions $\ORt(x)$, $\ORtp(x)$, $\nand(x)$ and $h_3(x)$ and their standard deviations for a photonic implementation of NMQC using four-photon GHZ states~\cite{Demirel2021}. The red (diagonally striped \textbackslash \textbackslash)/grey (dotted) bars denote the theoretically achievable quantum/classical bounds and the white (diagonally striped //) bars stand for the measured values.
	(b) and (c) Difference between the photonic results and the results presented in Sec.~\ref{sec:resultsC}, with and without error mitigation, i.e. $\Delta\beta_{\text{meas/max}} = \beta_{\text{meas/max}}(\text{photons}) - \beta_{\text{meas/max}}(\text{QSO})$ (white (diagonally striped //) bars) and $\Delta\beta_{\text{meas/max}}^{\text{corr}} = \beta_{\text{meas/max}}^{\text{corr}}(\text{photons}) - \beta_{\text{meas/max}}^{\text{corr}}(\text{QSO})$ (orange (plain) bars). (b) shows the difference to the averaged measured bounds. Without error mitigation the photonic values are always larger and $\Delta\beta_{\text{meas}}$ is positive. With error mitigation the differences become smaller, but only for OR$_3$ the mitigated result exceeds the photonic value. (c) shows the difference to the highest values produced by a single configuration. Even without error mitigation the differences $\Delta\beta_{\text{max}}$ are small, however only for OR$_3$ the photonic value is smaller and $\Delta\beta_{\text{max}}$ becomes negative. With error mitigation the results produced by QSO exceed the photonic values and $\Delta\beta_{\text{max}}^{\text{corr}}$ is always negative.}
 	\label{fig:Photons}
 \end{figure*}

\section{Conclusion and Outlook}\label{sec:concl}
 
On average, we have reached violations of all measured Bell inequalities for all tested functions listed in section~\ref{sec:BellInequ} on the 27-qubit IBM QSO in Ehningen, Germany. In the cases where Qiskit has chosen exactely one configuration of physical qubits for every experiment violations have been measured for all functions for up to \textit{six} qubits.
In contrast, in the cases where we have tested all possible $l$-qubit configurations, $l$ being the required number of qubits, and averaged over all results, the averaged measured bounds clearly violate the tested Bell inequalities (with and without error mitigation) for up to \textit{seven} qubits.
Since we have not only investigated a single qubit configuration, but averaged over all possible configurations, we have thereby tested the quantumness of the device. This means that we have demonstrated a computational advantage in terms of NMQC using the device IBM QSO and thus its non-local behaviour for up to seven qubits. Further we have compared our results using four qubits to an existing implementation of NMQC using four-photon GHZ states~\cite{Demirel2021}.

To improve the results and carry out NMQC for $h_k(x)$ for $k>7$ one could apply more sophisticated error mitigation/correction techniques~\cite{Mooney2021}. It would also be interesting to find other functions that translate to convenient Bell inequalities to test the non-classicality of quantum computers using this computational test. For this one could use the relation between NMQC and Reed-Muller codes hinted at in Sec.~\ref{sec:BellInequ}. However, as the performance of the qubits varies widely over time it should definitely be taken into account in order to obtain larger GHZ state fidelities and thus better results.
It is likely that in the future, more sophisticated qubit mapping methods will be developed, such as~\cite{Gerard2021}, which in combination with error mitigation and error correction methods could facilitate NMQC with large numbers of qubits.

Another possibility to reduce errors and noise in the generation of the GHZ states could be to minimise the depth of the quantum circuit. In~\cite{Piroli2021} a method has been discussed in which GHZ states of arbitrary size can be generated with constant circuit depth. Although additional ancilla qubits are needed here the advantage gained from the constant circuit depth would presumably beat possible problems caused by the increased number of qubits. From the generation of linear graph states on IBM QCs, which also has a constant depth, it is known that entangled states of much larger size can be generated~\cite{Mooney2021AQT,Yang2022}.

In conclusion, we have implemented NMQC for up to seven qubits using a 27-qubit IBM QC. We have shown that the calculation of non-linear Boolean functions and the simultaneous violation of multipartite Bell inequalities can be used to characterise quantum devices. This method can easily be extended to different quantum computing systems with qubits but also to higher-dimensional systems~\cite{Mackeprang2022}.

\begin{acknowledgments}
\begin{CJK}{UTF8}{mj}
We thank Jonas Helsen for the useful suggestions and explanations, Alexandra R. van den Berg for the fruitful discussions, Roeland Wiersema for the helpful tips and Chewon Cho (조채원) for the in-depth explanations.\end{CJK}~We thank Lukas Rückle and Christopher Thalacker for the explanations.
We acknowledge support from the Carl Zeiss Foundation, the Centre for Integrated Quantum Science and Technology (IQ$^\text{ST}$), 
the Federal Ministry of Education and Research (BMBF, projects SiSiQ and PhotonQ), 
the Federal Ministry for Economic Affairs and Climate Action (BMWK, project PlanQK), 
and the Competence Center Quantum Computing Baden-W\"urttemberg (funded by the Ministerium für Wirtschaft, Arbeit und Tourismus Baden-W\"urttemberg, project QORA).
We acknowledge the use of IBM Quantum services for this work. The views expressed are those of the authors, and do not reflect the official policy or position of IBM or the IBM Quantum team.

\end{acknowledgments}
 

%


\appendix

\begin{widetext}

\section{Quantum readout error mitigation}
\label{sec:QREM}

In this section we explain the details of the quantum readout error mitigation (QREM) introduced in Ref.~\cite{Maciejewski2020}. It aims at mitigating readout errors, which are errors during the measurement of the state of a single qubit. For example, a qubit might actually be in the state $\ket{1}$, but the measurement device asserts that it is in the state $\ket{0}$. The main assumption in QREM is that these measurement errors are \emph{local}.
This means that the measurement errors act on the probability vector $\vec{p}\equiv (p(0,0,0,...,0), p(0,0,0,...,1),...,p(1,1,1,...,1))^T$, where $p(m_{n-1},m_{n-2},...,m_1,m_0)$ is the probability of obtaining the measurement result $m_i$ for the measurement (in the computational basis) of the $i$th qubit $q_0$ ($i\in\{0,1,\ldots,n-1\}$), in the following way:
\begin{equation}
	\vec{p}\,' = \bigotimes_{i=1}^{n}{A}_{n-i} \vec{p}, \quad A_i \equiv \begin{pmatrix}
		p_i(0|0) & p_i(0|1) \\
		p_i(1|0) & p_i(1|1)
	\end{pmatrix}.
\end{equation}
The $A_i$ are called the \emph{calibration matrices} and $p_i(x|y)$ are the probabilities of measuring the state $x$ given that the $i$th qubit was actually prepared in the state $y$.

In order to mitigate the readout errors, one has to first compute the calibration matrices by preparing the qubits in the various states and then estimating the probabilities $p_i(x|y)$ using the law of large numbers. The corrected probability vector $\vec{p}$ is then obtained from the experimental probability vector $\vec{p}\,'$ by inverting the calibration matrices. However, sometimes, as the estimation of $A_i$ is not exact, the resulting $\vec{p}$ may not be an actual physical probability vector, meaning that some element of it may be smaller than 0 or all of them do not sum up to 1. Therefore, if that is the case, we use an optimisation method to find the closest physical probability vector $\vec{p}\,^{*}$ to $\vec{p}$. To be exact, $\vec{p}\,^{*}$ is given by~\cite{Maciejewski2020}:
\begin{equation}
	\vec{p}\,^{*}=\argmin_{\forall i \, \tilde{p}_i \geq 0, \sum_i \tilde{p}_i = 1}(\lVert\tilde{\vec{p}}-\vec{p}\rVert),
\end{equation}
where $\lVert\cdot\rVert$ is the euclidean norm.

\section{Additional proofs}\label{app:proofs}
 
 \subsection{Proof of efficient computability of \texorpdfstring{$h_k(x)$}{hk(x)}  for all \texorpdfstring{$k$}{k}}
 
 Here, we prove the following:
 
\begin{theorem}\label{thm:hkquantum}
	To deterministically compute the $k$-bit function $h_k(x)$, as defined in~\eqref{eq:hk}, one only requires $l=k+1$ qubits forming the $l$-qubit generalised GHZ state (see~\eqref{eq:2GHZ}), the pre-processing defined in Eqn.~\eqref{eq:preprohk} and the measurement settings~$\hat{m}_i(s_i=0)=X$, $\hat{m}_i(s_i=1)=Y$ $\forall i$.
\end{theorem}

\begin{proof}
	
This can easily be proven via natural induction. We start by observing the following:
	\begin{align}
		h_{k+1}(x)&=\bigoplus_{i=0}^{k-1} \bigoplus_{j=i+1}^{k}x_i x_j\oplus \bigoplus_{i=0}^{k} x_i \\
		&=\bigoplus_{i=0}^{k-2} \bigoplus_{j=i+1}^{k-1}x_i x_j \oplus x_{k}\left(\bigoplus_{i=0}^{k-1} x_i\right) \oplus \bigoplus_{i=0}^{k-1} x_i \oplus x_k \\
		&=h_k(x)\oplus x_k\left(\bigoplus_{i=0}^{k-1} x_i\right) \oplus x_k, \label{eq:hkp1}
	\end{align}
meaning that we can write $h_{k+1}(x)$ as a sum of $h_k(x)$ and a term that is only dependent of $x_k$ and the sum over the first $k$ bits of $x$, i.e. $x_0$, $x_1$,$\ldots$, $x_{k-1}$. We abbreviate this sum by $S$:
\begin{equation}\label{eq:S}
	S= \bigoplus_{i=0}^{k-1} x_i .
\end{equation}

We will now start with the actual proof by induction. For this, we need the general condition for NMQC to be successful. In~\cite{Hoban2011IOP}, it is shown that any Boolean function can be computed with NMQC when using the generalised GHZ state and the general measurement operators:
\begin{equation}
	\hat{m}_i(s_i)=\cos(\phi_i s_i)X+\sin(\phi_i s_i)Y.
\end{equation}

The condition that must be fulfilled for the deterministic computation of a function $f: \{0,1\}^n \rightarrow \{0,1\}$ is~\cite{Hoban2011IOP}:
\begin{equation}\label{eq:NMQCcond}
	e^{i\sum_{j=0}^{l-1}s_j\phi_j} = (-1)^{f(x)+c},
\end{equation}
where $\phi_j \in (-\pi,\pi)$ are angles yet to be determined, $c$ is a bit that can be added in post-processing and $s_j$ is related to the input bit string $x$ by the pre-processing $s=(Px)_{\oplus}$. To prove theorem~\ref{thm:hkquantum} by natural induction, we first show that condition~\eqref{eq:NMQCcond} holds for $h_3(x)$ and then show that it is fulfilled when $k\mapsto k+1$. As $h_k(x)$ is zero for $x=0$ for all $k$, we can set $c$ in Eqn.~\eqref{eq:NMQCcond} to zero. Additionally, we know that $\phi_j=\pi/2$ $\forall j$ due to the fixed measurement settings. This, combined with the pre-processing given by Eqn.~\eqref{eq:preprohk} simplifies Eqn.~\eqref{eq:NMQCcond} for $h_k(x)$ to:
\begin{equation}\label{eq:NMQCcondproof}
	e^{i\frac{\pi}{2}\left(\sum_{j=0}^{k-1}x_j + S\right)} = (-1)^{h_k(x)},
\end{equation}
where $S$ is the abbreviation for the sum modulo 2 of all $k$ bits in $x$ [see Eqn.~\eqref{eq:S}].
\begin{enumerate}
	\item $k=3$. We list the values for $\sum_{j=0}^{k-1}x_j$, $S$, $h_3(x)$ and $e^{i\frac{\pi}{2}\left(\sum_{j=0}^{k-1}x_j + S\right)}$ for all 8 input strings in Table~\ref{tab:IH}, of which one can read of that the induction hypothesis is fulfilled.
	\item $k\mapsto k+1$. When increasing $k$ by 1, the r.h.s of Eqn.\eqref{eq:NMQCcondproof} becomes:
	\begin{equation}\label{eq:hkp1ind}
		(-1)^{h_{k+1}(x)}=(-1)^{(h_k(x)\oplus x_k S \oplus x_k)},
	\end{equation}
where we have inserted Eqn.~\eqref{eq:hkp1}. According to condition~\eqref{eq:NMQCcondproof}, for deterministic NMQC to function for $k\mapsto k+1$, the following must hold: 
\begin{align}
		(-1)^{(h_k(x)\oplus x_k S \oplus x_k)}&=(-1)^{h_k(x)}\cdot (-1)^{x_k S \oplus x_k)} \\
		&= e^{i\frac{\pi}{2}\left(\sum_{j=0}^{k-1}x_j + S\right)}  \cdot  (-1)^{x_k S \oplus x_k} \label{eq:insertIH}\\
		&=e^{i\frac{\pi}{2}\left(\sum_{j=0}^{k}x_j  + (S\oplus x_k) \right)}.
		\end{align}
	
	We have inserted the induction hypothesis~\eqref{eq:NMQCcondproof} in Eqn.~\eqref{eq:insertIH}. In summary, we must show that this:
	\begin{equation}
		 e^{i\frac{\pi}{2}\left(\sum_{j=0}^{k-1}x_j + S\right)}  \cdot  (-1)^{x_k S \oplus x_k} = e^{i\frac{\pi}{2}\left(\sum_{j=0}^{k}x_j  + (S\oplus x_k) \right)}
	\end{equation}
is true for all $x$. We can cancel out  $e^{i\frac{\pi}{2}\left(\sum_{j=0}^{k-1}x_j \right)}$ on both sides and are left with:
 
 	\begin{equation}\label{eq:indlastcond}
 	e^{i\frac{\pi}{2} S}  \cdot  (-1)^{x_k S \oplus x_k} = e^{i\frac{\pi}{2}\left(x_k+(S\oplus x_k)\right)}.
 \end{equation}
Eqn.~\eqref{eq:indlastcond} only depends on $x_k$ and $S$. The final steps of this proof thus merely consist of checking if it is correct for the four possible combinations of $x_k$ and $S$. For both $x_k=0$ and $S=0$, both sides of the equation are equal to 1. For $x_k=1$ and $S=0$, the l.h.s. becomes ($-1$) and the r.h.s becomes $e^{i\frac{\pi}{2}(1+1)}=e^{i\pi}$. For $x_k=0$ and $S=1$, the l.h.s. equals $e^{i\frac{\pi}{2}}$ and the r.h.s. turns into $e^{i\frac{\pi}{2}\left(0+1\right)}$. Lastly, for $x_k=S=1$, the l.h.s. is $i$ and the r.h.s. is equal to $e^{i\frac{\pi}{2}\left(1+0\right)}=i$. 
\end{enumerate}
This completes the proof of theorem~\ref{thm:hkquantum}.

\begin{table}[b]
	\centering
	\begin{tabular}{c|c |c |c|c|c|c|c|c} 
	$x$ & (0,0,0) & (0,0,1) & (0,1,0)& (0,1,1) & (1,0,0) & (1,0,1) & (1,1,0)& (1,1,1) \\ 	\hline
		$\sum_{j=0}^{k-1}x_j$ & 0 & 1 & 1 &2  & 1 & 2 & 2 &3 \\
		\hline 
		$S$ & 0 & 1 & 1 & 0 & 1 & 0 & 0 & 1 \\ \hline
		$h_3(x)$ & 0 &1 & 1 & 1 & 1 & 1 & 1 & 0 \\ \hline
		$e^{i\frac{\pi}{2}\left(\sum_{j=0}^{k-1}x_j + S\right)}$ & $e^{0}=1$ & $e^{i\frac{\pi}{2}\cdot 2}=-1$&$-1$ &  $e^{i\frac{\pi}{2}(2+0)}=-1$&$-1$ &$-1$&$-1$& $e^{i\frac{\pi}{2}\left(3+1\right)}=1$
		\end{tabular}
	\caption{Table to test induction hypothesis for $h_3(x)$}\label{tab:IH}
\end{table}
\end{proof}

Note that, using the appropriate measurement settings and pre-processing, the NMQC output will always be $h_k(x)$, which means that any probabilistic NMQC game using the same measurement settings and pre-processing (and the GHZ state as a computational resource) has a quantum success probability of 1, translating to a bound $q=1$ of the associated Bell inequality.

 \subsection{Classical bound of the probabilistic NMQC game \texorpdfstring{$h_k(x)$}{hk(x)} for all \texorpdfstring{$k$}{k}}

Here, we explain how to obtain the classical success probabilities of the probabilistic NMQC games induced by $h_k(x)$ with a uniform sampling distribution $\xi(x)=1/2^k$. To be precise, we prove the following theorem:
 
 \begin{theorem}\label{thm:hkclassical}
 	The LHV bound $c$ of the Bell inequality bounding the average success probability of the probabilistic NMQC game [according to Eqn.\eqref{eq:psu}] induced by the function $h_k(x)$ defined in Eqn.~\eqref{eq:hk} with a uniform sampling distribution $\xi(x)=1/2^k$ is equal to $c=2^{-\frac{k}{2}}$ for even $k$ and $c=2^{-\left(\frac{k-1}{2}\right)}$ for odd $k$.
 \end{theorem}
\begin{proof}
	The proof entirely consists of combining one simple observation with previous knowledge on the non-linearity of Boolean functions. As, by convexity, the best LHV strategy in an NMQC game with a uniform sampling distribution $\xi(x)=1/2^n$, where the goal is to compute a Boolean function $f:\{0,1\}^n\rightarrow \{0,1\}$, is to output the \emph{closest} linear function, the average classical success probability $\psc$ is given by:
\begin{equation}\label{eq:pscNL}
		\psc=\left[2^n-\min_{g\,\text{linear}}\mathrm{dist}(f,g)\right]/2^n,
\end{equation}
where the distance $\mathrm{dist}(f,g)$ of two Boolean functions $f$ and $g$ is their Hamming distance, i.e. the number of arguments $x$ for which $f(x)\neq g(x)$. \newline

As already pointed out in~\cite{Raussendorf_2013}, the minimum distance of a function $f$ to the closest linear function $g$ is its \emph{nonlinearity} $\NL$. Boolean functions with the maximum nonlinearity possible are called \emph{bent} functions  ~\cite{Savicky1994,Maitra2002}. They lead to a minimal classical success probability $\psc$ and are thus best suited to demonstrate non-locality with this type of NMQC game (see also~\cite{Raussendorf_2013}). \newline
The function $h_k(x)$, as defined by Eqn.~\eqref{eq:hk} is in fact one of the few symmetric bent functions~\cite{Savicky1994,Maitra2002}. For even $n$, its non-linearity is $\NL =2^{n-1}-2^{n/2-1}$\cite{Savicky1994}. For odd $n$, its non-linearity is $\NL=2^{n-1}-2^{(n-1)/2}$. Inserting these values into Eqn.~\eqref{eq:pscNL} and using Eqn.~\eqref{eq:psu}, one then immediately obtains the classical bounds of the associated Bell inequalities.

\end{proof}

Note that, as the quantum bound of the Bell inequality $q$ is always 1, this NMQC game is related to a Bell inequality, for which the ratio $q/c$ between the quantum bound and its LHV counterpart increases exponentially.

\section{Individual measured bounds}
\label{sec:IndividualMeasuredBounds}

To understand why the averaged results are notably smaller than the maximal ones produced by a single configuration we will take a look at all individual measured bounds, i.e. all specific qubit configurations, of two Bell inequalities. Fig.~\ref{fig:BetasindividualOR3Plus} shows the (mitigated) bounds of the Bell inequality induced by $\ORtp(x)$ (see Table~\ref{tab:Inequalities}) for every single four-qubit configuration. One can see that the measured bounds strongly vary, even including negative values. This coincides with the measured expectation values for a qubit configuration, which produced high values (0-1-2-3), and a qubit configuration, which produced negative values (10-18-12-15) (see Fig.~\ref{fig:expvaluesOR3p}).

In general, the performance of single qubits varies over time. The experiment for $\ORtp(x)$ was run on the 7th June 2022, whereas the one for $h_3(x)$ was run on 21st May 2022. At the time of the four outlier NMQC runs for $\ORtp(x)$, the readout-error rates of the qubits 18 and 12 were 0.011 and 0.022 compared to 0.017 and 0.008 for the same qubits during the measurements belonging to the probabilistic NMQC game induced by $h_3(x)$. We plot the violations of the associated Bell inequality for $h_3(x)$ in Fig.~\ref{fig:BetasindividualH3}, where no qubit configuration exhibits this kind of behaviour. Therefore, to show the violation of Bell inequality for a single qubit configuration one has to perform multiple runs at different times and average the results.

\begin{figure*}[t]
	\centering
	\includegraphics[width=1.0 \linewidth]{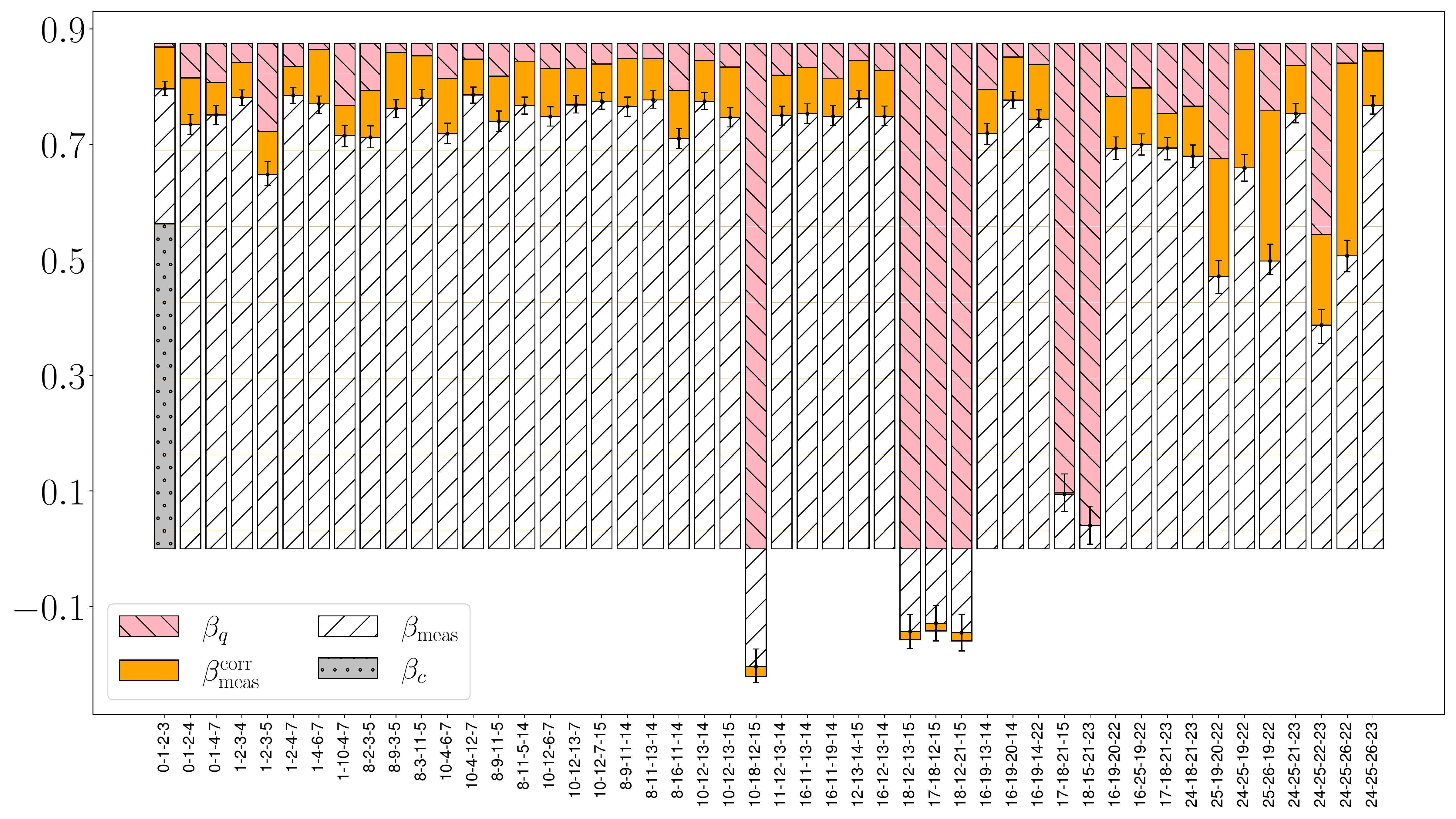}
	\caption{(Mitigated) measured bounds of the Bell inequality for every qubit configuration induced by the functions $\ORtp(x)$.}
	\label{fig:BetasindividualOR3Plus}
\end{figure*}

\begin{figure*}[t]
	\centering
	\includegraphics[width=1.0 \linewidth]{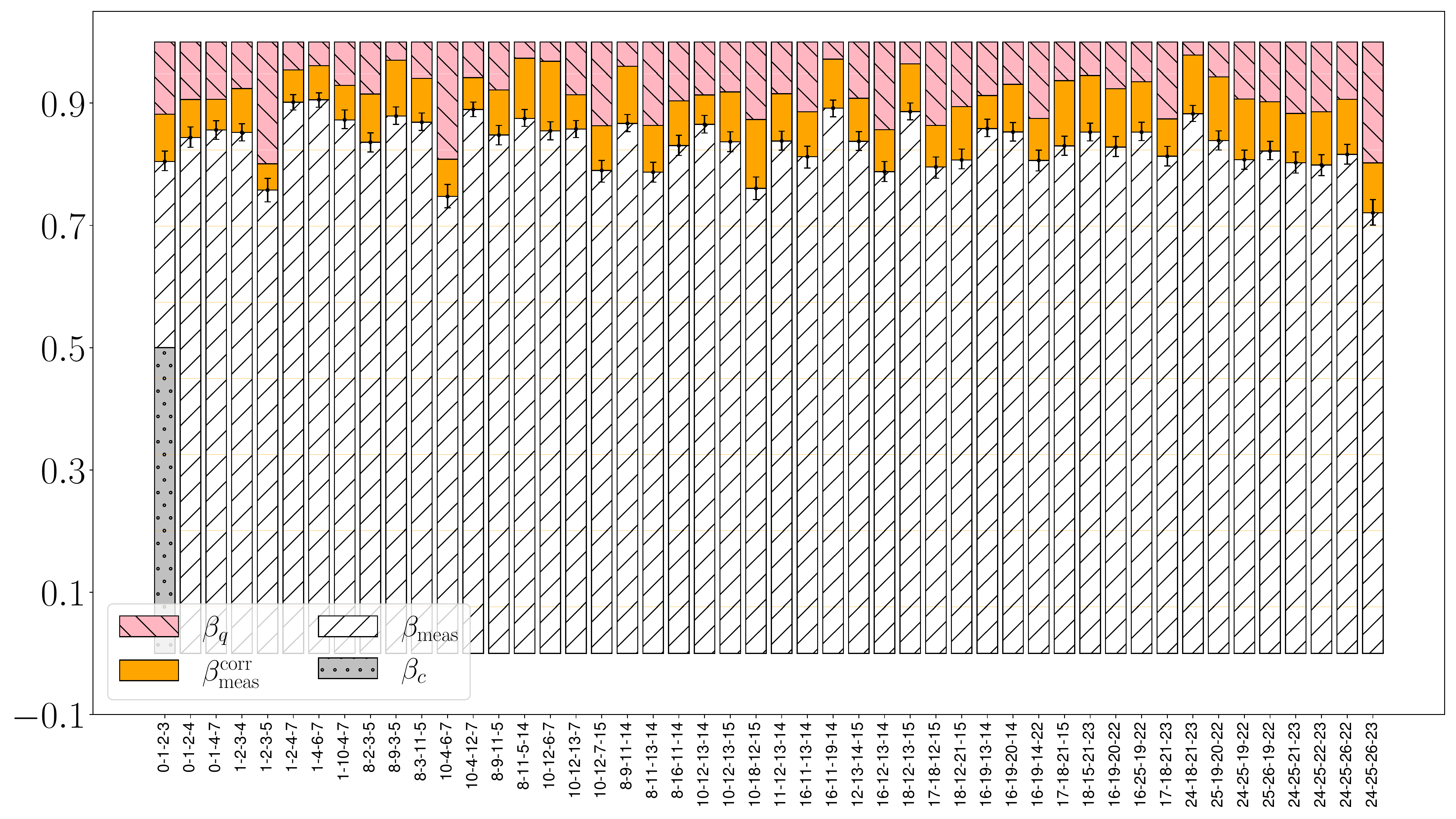}
	\caption{(Mitigated) measured bounds of the Bell inequality for every qubit configuration induced by the functions $h_{3}(x)$.}
	\label{fig:BetasindividualH3}
\end{figure*}

 \begin{figure*}[t]
	\centering
	\includegraphics[width=1.0 \linewidth]{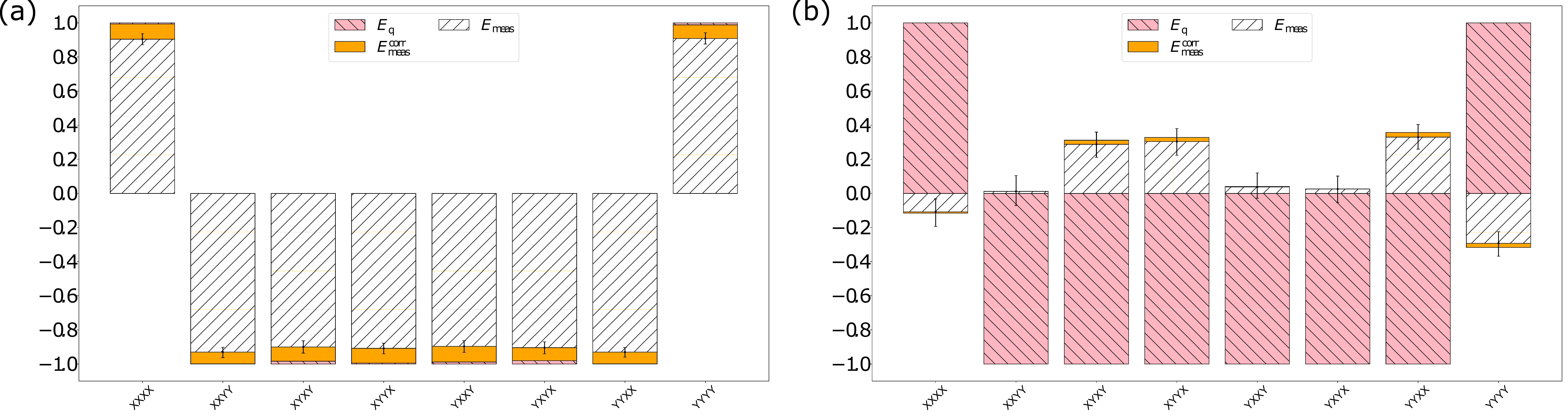}
	\caption{Individual expectation values of the operators making up the Bell operator induced by the probabilistic NMQC game for the function $\mathrm{OR}_3^{\oplus}$ for two different qubit configurations. In theory, they should all be $\pm 1$. (a) Qubit configuration 0-1-2-3. (b) Qubit configuration 10-18-12-15.}
	\label{fig:expvaluesOR3p}
\end{figure*}

\newpage
\end{widetext}

\end{document}